\documentclass{article}
\usepackage{amsmath,amssymb,amsthm}
\usepackage{graphics,graphicx,pifont,tikz}
\usepackage{url} 

\newcommand{\comment}[1]{}

\theoremstyle{plain}
\newtheorem{theorem}{Theorem}[section]

\newtheorem{lemma}[theorem]{Lemma}

\theoremstyle{definition}
\newtheorem{definition}[theorem]{Definition}

\theoremstyle{remark}
\newtheorem{remark}{Remark}

\def\eps{{\varepsilon}}
\def\l{{\lambda}}
\def\Var{{Var}}
\def\ML{{Matching Level}}
\def\ml{{matching level}}
\def\MLs{{Matching Levels}}
\def\mls{{matching levels}}

\newcommand{\maxint}{\textsc{MaxInt}}
\newcommand{\mirank}{\textsc{MiRank}}
\newcommand{\calG}{\mathcal{G}}
\newcommand{\calP}{\mathcal{P}}
\newcommand{\calF}{\mathcal{F}}
\newcommand{\calD}{\mathcal{D}}
\newcommand{\calT}{\mathcal{T}}
\newcommand{\calO}{\mathcal{O}}

\newcommand{\new}{\mathit{new}}
\newcommand{\poly}{\mathit{poly}}

\newcommand{\PN}[1]{\textbf{#1}\par}

\newcommand{\DB}{\item[Database]}

\newcommand{\QY}{\item[Query]}

\newcommand{\CONSTR}{\item[Constraints]}

\newenvironment{BoxItList}%
  {\begin{list}{}%
         {%
           \setlength{\itemsep}{2pt}%
         }%
  }%
  {\end{list}}
\newsavebox{\savefboxit}
\newlength{\savefboxitsep}
\newlength{\fboxitlength}
\newenvironment{BoxIt}[1][Consider:]{%
  \setlength{\savefboxitsep}{\the\fboxsep}%
  \setlength{\fboxsep}{3mm}
  \setlength{\fboxitlength}{11.5cm} 
  \begin{lrbox}{\savefboxit}
  \begin{minipage}{\fboxitlength}
  \PN{#1}
  \begin{BoxItList}
  }
  {
   \end{BoxItList}
   \end{minipage}
   \end{lrbox}\noindent\begin{center}\fbox{\usebox{\savefboxit}}%
   \end{center}
  }

\newenvironment{BoxItAlgo}{%
  \setlength{\savefboxitsep}{\the\fboxsep}%
  \setlength{\fboxsep}{3mm}
  \setlength{\fboxitlength}{11.5cm} 
  \begin{lrbox}{\savefboxit}
  \begin{minipage}{\fboxitlength}
  \begin{BoxItList}
  }
  {
   \end{BoxItList}
   \end{minipage}
   \end{lrbox}\noindent\begin{center}\fbox{\usebox{\savefboxit}}%
   \end{center}
  }

\begin{document}

\title{Maximal Intersection Queries in Randomized Input Models}
\author{
  Benjamin Hoffmann\thanks{Email: hoffmann@fmi.uni-stuttgart.de}\\{\small
  Universit\"at Stuttgart, Germany}
  \and
  Mikhail Lifshits\thanks{Email: lifts@mail.rcom.ru}\\
  {\small St.~Petersburg State University, Russia}
  \and
  Yury Lifshits\thanks{Email: yury@caltech.edu}\\
  {\small California Institute of Technology, USA}
  \and
  Dirk Nowotka\thanks{Email: nowotka@fmi.uni-stuttgart.de}\\{\small
  Universit\"at Stuttgart, Germany}
}
\date{}

\maketitle

\begin{abstract}
  Consider a family of sets and a single set, called the
  query set. How can one quickly find a member
  of the family which has a maximal intersection with
  the query set? Time constraints on the query and on a possible
  preprocessing of the set family make this problem challenging.
  Such maximal intersection queries arise in a wide range of applications,
  including web search, recommendation systems, and distributing
  on-line advertisements. In general, maximal intersection
  queries are computationally expensive. We investigate two 
  well-motivated distributions over all families
  of sets and propose an algorithm for each of them. We show that with
  very high probability an almost optimal solution is found in
  time which is logarithmic in the size of the family. Moreover, we point
  out a~threshold phenomenon on the probabilities of intersecting
  sets in each of our two input models which leads
  to the efficient algorithms mentioned above.
\end{abstract}

\section{Introduction}\label{sec:intro}
The \emph{nearest neighbor} problem is the task to determine in a
general metric space a point that is closest to a given query
point. This kind of queries appear in a~huge number of applied
problems: text classification, handwriting recognition,
recommendation systems, distributing on-line advertisements,
near-duplicate detection, and code plagiarism detection.

In this paper we consider the nearest neighbor problem in a ``binary''
form. Namely, every object is described as a set of its features
and similarity is defined as the number of common features. 
\comment{
In this paper we consider the nearest neighbor problem in a ``binary''
form. Namely, every object is described as a set of its features
and similarity is defined as the number of common features. For
some cases, like recommending a person who has a maximal number of
joint friends with you but is not your direct friend, this
formalization is quite natural. Moreover, weighted models
could be simply reduced to the binary form. Let us illustrate this
reduction with an example. Assume that we are working with documents
and their terms, and one particular term is ``Ekaterinburg''. Then
we can introduce eight new artificial terms Ekaterinburg1, \dots
Ekaterinburg8. For an object having the largest weight for
Ekaterinburg in weighted representation we simply put all eight
new terms, while for objects with small weights we put only
Ekaterinburg1. Thus, an intersection similarity for the new
representation is somehow reflecting the classical scalar product
similarity for the vector model.
} In order to construct an efficient solution some assumptions
should be added to the problem. Here we assume that the
input behaves according to some predefined distribution. Then we
construct an algorithm and show that the time complexity and/or
the accuracy are reasonably good \emph{with high probability}.
Here we use the probability over the input distribution, not over
random choices of the algorithm. This probabilistic approach was
inspired by the recent survey of Newman~\cite{N03}. He gives a
comprehensive survey about random models of graphs that agree well
with many real life networks, including Web graphs, friendship
graphs, co-authorship graphs, and many others. Hence, we can
attack the nearest neighbor problem in already ``verified'' random models.

\paragraph{The Maximal Intersection Problem.} Consider a family
of sets and a single set. We ask for a member of the set family
which has a maximal intersection with the query set.

\begin{BoxIt}[The Maximal Intersection Problem (\maxint)]
  \DB A family $\calF$ of $n$ sets such that $|f|\leq k$
    for all $f\in \calF$.
  \QY Given a set $f_\new$ with $|f_\new|\leq k$,
    return $f_i\in\calF$ with maximal $|f_\new\cap f_i|$.
  \CONSTR Preprocessing time $n\cdot (\log n)^{\calO(1)}\cdot k^{\calO(1)}$.\\
    Query time $(\log n)^{\calO(1)}\cdot k^{\calO(1)}$.
\end{BoxIt}

Let us restate the problem in a graph theoretical
notation which will allow a more convenient description 
of some applications of \maxint\ later. A~{data\-base} is a
bipartite graph with vertex set partition $(V, V')$ such that
$|V|=n$ and the degree of every $v\in V$ is at most $k$. A query
is a (new) vertex $v$ (together with edges connecting $v$ with
$V'$) of degree at most $k$. The query task is to return a vertex
$u\in V$ with a maximal number of paths of length $2$ from $v$ to
$u$.

Our main motivation for studying \maxint\ was problems like 
text clustering, near-duplicate detection or distribution of
on-line advertisements. In these problems, the database mainly consists of 
natural language text documents. Therefore, we will deal in the rest of the paper with
documents and terms instead of sets and elements. On the other hand,
we want to stress that our ideas and algorithms can be applied to
every input following our models. Note that in this work documents are not 
considered to be multisets of terms. But, as we will see in Section \ref{sec:ZipfModel}, 
we use the fact that every term in a document occurs with a certain multiplicity.

\paragraph{Results.}
In Section 2 and Section 3 we propose two new randomized input models for \maxint, 
called the \emph{Zipf model} and the \emph{hierarchical scheme}. Assume that the terms of a query document are
ordered by their frequency in the document collection. Now
consider the probability curves for the two following events
with parameter $q$ (Figure 1). \emph{Any $q$-match}: there is a document in
the random (according to our models) collection that has at least
$q$ common terms with the query document (the solid curve).
\emph{Prefix $q$-match}: there is a document in the random collection that
has at least the first $q$ terms (according to the order given by the term frequencies) 
of the query document (the dashed curve).
Both curves have the similar structure: the probability is close to 1 for small $q$,
but suddenly, at some ``\ml'', it falls to nearly zero.
\begin{figure}[h]
\includegraphics[scale=1]{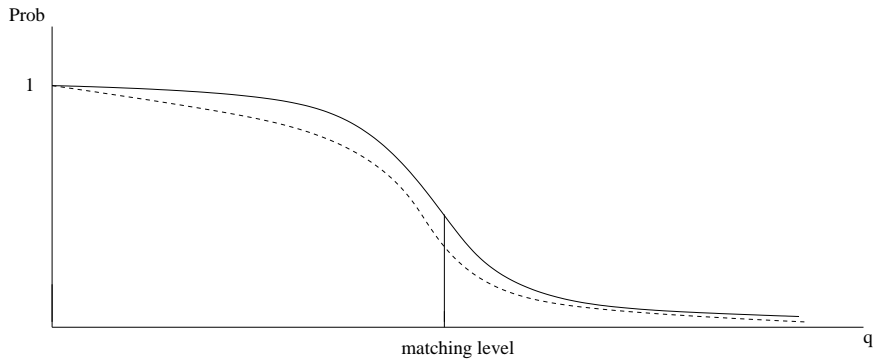}
\caption{Exemplary probability curves for any $q$-match and prefix $q$-match}
\end{figure}
Our main observation is that these \mls\ for prefix $q$-match and any $q$-match are very close to each
other. And this is extremely important for solving \maxint.
Indeed, finding the best prefix $q$-match is computationally feasible.
We show that closeness of \mls\ for any $q$-match and prefix $q$-match with high
probability allows to find an approximate solution for \maxint.

With respect to the conference version of this paper we generalize the \ml\ theorem 
from regular queries to any query taken from the Zipf model \cite{HLN07}.

\paragraph{Applications.} The \maxint\ problem is a natural
formalization for many practical problems:
\begin{description}
  \item[Long search queries.]
    Consider a bipartite graph representing websites and their
    words, that is, every website is represented as a set of words.
    Let a query be a moderately large set of words (say, 100 words).
    For example, one might get a large query by expanding a five word query
    by adding all
    synonyms. The search for a website containing ALL query terms
    might produce no result. Hence, a search for a website that has
    maximal intersection with the query set is a natural alternative.
    Therefore, efficient algorithms for \maxint\ can help web search
    engines like \url{google.com}\footnote{The reference time for all links
    mentioned in this paper is April 2007.}, \url{yahoo.com},
    or \url{msn.com} to relax their restrictions on the query length.
  \item[Content-based similarity.]
    Consider a bipartite graph representing documents and their terms.
    Finding a document in a database that has a maximal number of
    common terms with a newcomer document
    might be a basic routine for text clustering/classification and
    duplicates detection. Particular examples are news classification
    systems (\url{reuters.com}), news clustering (\url{news.google.com}),
    and spam detection.
  \item[New connection suggestions.]
    Consider an undirected graph between people representing for example
    friendship or co-authorship. Here, every person is described by his name
    and a list of all his friends. Then, applying a \maxint\ query to
    a (new) person we get a natural suggestion for establishing
    a new connection for her. Indeed, we get a person that has
    a maximal number of joint friends with the query person.
    Related systems can be found at \url{linkedin.com} (for acquaintances)
    and \url{dblp.uni-trier.de} (for co-authorship).
  \item[Co-occurrence similarity.]
    Consider an audience graph between people and some items.
    Every item is represented by a set of people who are interested in
    it. Take a (new) item together with its audience. Then, the
    \maxint\ query returns an item that has the maximal co-occurrence
    with the query item in people's preference lists. Particular examples
    are the music band similarity by their listeners
    (\url{last.fm}) and RSS-feeds similarity by their subscribers
    (\url{bloglines.com} and \url{feedburner.com}).
  \item[Advertisement Matching.]
    Delivering advertisement  relevant to users interests is one of the most
    important problems in web technologies \cite{LN07}. \maxint\ can
    reflect this challenge in a natural way: Consider a graph representing
    websites participating in some ad distribution system and their terms. A query is a set of
    terms that describe some advertisement and its target audience.
    Here, a solution of \maxint\ suggests a website that is among the best candidates
    to display the given ad. See \url{google.com/adsense} as an example
    system for ad distribution.
  \item[Social recommendations.]
    Consider a bipartite graph between people and their recommendations
    (e.g. for books, bars, cars). A query is a set of friends of some
    newcomer person. Finding an item that is already chosen
    by many of the newcomer's friends is a natural form of recommendation.
    The friendship graph together with recommended items is for example
    accumulated on \url{facebook.com}.
\end{description}

Note that for some of these applications the Jaccard similarity coefficient may also be 
an appropriate similarity measure.

\paragraph{Related Work.} \maxint\ is a special case of the nearest
neighbor problem. Indeed, one just needs to define the similarity
between two documents as the number of common words. There is also
a way to define a {\it metric} (i.e. distance function satisfying
the triangle inequality) providing the reverse similarity order. To do
this, we need to add some unique ``imaginary'' words to every
document making their size equal and then use the Jaccard metric
\cite{B97}. Denoting the maximal cardinality of a document in the
collection by $M$ the resulting formula is $d(A,B)=\frac{2M-2|A\cap B|}{2M-|A\cap B|}$. 
Instead of the Jaccard metric one could use the size of the symmetric difference as well (again, one has to 
add unique words to every document making their size equal). This defines again a metric which provides the 
reverse similarity order.

Many efficient algorithms have been developed for nearest
neighbor search in special cases or under various assumptions; see
recent survey papers \cite{BBK01,CNB+01,C06,HS03,I04} and the book
\cite{ZAD+06} for comprehensive reviews. Nearest neighbors are
particularly well studied in vector models with the Euclidean distance
function \cite{K97,FKS03}. Actually, we can interpret a~document
as a vector of $0$s and $1$s ($1$ means a term is contained in a
document). Then, the scalar product is equal to the size of the
intersection. Unfortunately, random projection methods studied
are not directly applicable to \maxint. Namely, (1) we do not
allow that the complexity is linear in the vector length, and (2)
a $c$-approximate solution for the Euclidean distance is not necessarily
a $c$-approximate solution for the size-of-intersection similarity. 
Note that the length of the vectors (resulting from the 
overall number of different terms in the document collection) can be much 
larger than the size of the document collection. 
\comment{
For finding near-duplicates in the Web minhashing \cite{BCF+98}
and locality-sensitive hashing \cite{AI06} techniques were
suggested. In a nutshell, in the preprocessing stage every document is
hashed with several randomly chosen hash functions. Given a query
document one compute its hash functions and check explicitly all
documents having the same values. Actually, a variant of our algorithm can be
used as a hashing method where all combinations of regular documents are considered to 
be the hash value of a document. It is straightforward to see from 
Theorem 2.3 that a prefix match of a regular document occurs with 
high probability.
}

Closely related to \maxint\  is \emph{text search}. Finding
documents that fit best to some given search terms can also be
considered as a problem on a bipartite graph. The documents and
terms are the nodes and edges are drawn when a term occurs in a
document. Basically the task is to find all documents containing
\emph{every} query term and rank these documents by relevance. The
key technique in this area is inverted files (inverted indexing).
A comprehensive survey of the topic can be found in \cite{ZM06}.

\section{\maxint\ in the Zipf Model}\label{sec:ZipfModel}

\comment{
Throughout the following sections we use a documents-terms notation,
that is, elements from $\calD = \{d_1,\ldots,d_n\}$ represent
documents and elements from $\calT = \{t_1,\ldots,t_m\}$ represent
terms. Let $m \leq \poly(n)$. By $\log$ we always
mean $\log_2$, while $\ln$ denotes $\log_e$.
}

Let $\calT=\{t_1,\ldots,t_m\}$ be a set of \emph{terms} and $\calD = \wp(\calT)$ be the power
set of $\calT$, called \emph{documents}.
A document collection $\calD_n$ is a subset $\{d_1,\ldots,d_n\} \subseteq\calD$.
We demand $m \in n^{\calO(1)}$. In the following we will use the terms \emph{prefix match} and 
\emph{any match} instead of prefix $q$-match and any $q$-match since the size of a matching 
is always stated explicitly. By $\log$ we always mean $\log_2$, while $\ln$ denotes $\log_e$.

We now describe a probabilistic mechanism for generating a
document collection called the \emph{Zipf model}. Every document
is generated independently. Term occurrences are also
independent. A document contains term $t_i$ with probability
$1/i$. Hence, the expected number of terms in a document
is approximately equal to $\ln m$ in our model. This model is
similar to the \emph{configuration model} (\cite{N03}) with
Zipf's law for distribution of term degrees and constant document
degrees. Zipf's law states that in natural language texts the
frequency $f$ of a word is approximately inversely proportional to its
rank $r$ in the frequency table, i.e. there exists a constant $c$ such that
$f \cdot r \approx c$ (Table 1)\comment{\footnote{ In the \emph{frequency
table}, the most frequent term is at rank 1, the second most
frequent term at rank 2 and so on.}}. For more details about Zipf's
law see \cite{MS02}.
\begin{table}[h]
\begin{center}
  \begin{tabular}{l@{\quad}r@{\quad}r@{\quad}r@{\quad}|@{\quad}l@{\quad}r@{\quad}r@{\quad}r}
  Word & Frequency & Rank & $f\cdot r$ & Word & Frequency & Rank & $f \cdot r$\\
  \hline
  the & 3332 & 1 & 3332 & turned & 51 & 200 & 10200\\
  and & 2972 & 2 & 5944 & you'll & 30 & 300 & 9000\\
  a & 1775 & 3 & 5235 & name & 21 & 400 & 8400\\
  he & 877 & 10 & 8770 & comes & 16 & 500 & 8000\\
  but & 410 & 20 & 8400 & group & 13 & 600 & 7800\\
  be & 294 & 30 & 8820 & lead & 11 & 700 & 7700\\
  there & 222 & 40 & 8880 & friends & 10 & 800 & 8000\\
  one & 172 & 50 & 8600 & begin & 9 & 900 & 8100\\
  about & 158 & 60 & 9480 & family & 8 & 1000 & 8000\\
  more & 138 & 70 & 9660 & brushed & 4 & 2000 & 8000\\
  never & 124 & 80 & 9920 & sins & 2 & 3000 & 6000\\
  Oh & 116 & 90 & 10440 & Could & 2 & 4000 & 8000\\
  two & 104 & 100 & 10400 & Applausive & 1 & 8000 & 8000
  \end{tabular}\\
\vspace*{0.2cm}
\caption{Empirical evaluation of Zipf's law on Tom Sawyer}
\end{center}
\end{table}
\begin{remark}The frequency of a term $t$ in a collection $\calD_n$ of
documents is defined as $$\frac{|\{d \in \calD_n\>|\>t \in
d\}|}{n}.$$ The expected frequency of the term $t_i$ is
equal to $1/i$. At the same time, the expected frequency
rank for $t_i$ is exactly the $i$-th value among those of all
terms. So the Zipf model reflects in a natural way Zipf's law.
Since some of our motivating applications also deal with natural
language texts, we can state that the Zipf model agrees with real
life at least by degree distribution.
\end{remark}
\begin{remark}
By defining the probability of the term $t_i$ to be contained in a document as $1/i$,
the set $\calD$ yields a probability space where a document $d$ is an event that occurs with 
probability $P(d) = \left(\prod_{t_i\in d}\frac{1}{i}\right) \left(\prod_{t_i\not\in d} 1-\frac{1}{i}\right)$. 

\end{remark}
In the following proofs we will use two inequalities ($a,b >0$):
\begin{equation}\label{eq:1}
\left(1 - \frac{a}{b}\right)^b < e^{-a},\qquad a \leq b
\end{equation}
\begin{equation}\label{eq:2}
\left(1 - \frac{1}{ab}\right)^{a} \geq 1-\frac{1}{b},\qquad a,b \geq 1
\end{equation}
Indeed, let $g(x)=\ln(1-x)/x,\> 0< x< 1$. Notice that $g$ is
a decreasing function. 
The first inequality follows from $g(a/b)\le \lim_{x\to 0} g(x)=-1$, while the
second one is equivalent to $g(1/b)\le g(1/ab)$.\\

For further considerations we introduce the following terms and definitions: 
\comment{
\begin{description}
\item[Regular document:]
Let $$\underbrace{t_1\;t_2}_{P_1}\quad\underbrace{t_3\;t_4\;t_5\;t_6\;t_7}_{P_2}\quad\ldots$$
be a partition of the set of terms. The group $P_i$ includes terms from $t_{\lceil e^{i-1} \rceil}$
to $t_{\lfloor e^i\rfloor}$. We say that a document $d\in \calD$ is \emph{regular} if it contains exactly $\ln m$ terms
$p_1\ldots p_{\ln m}$ such that $p_i \in P_i$.
\item[$\delta$-$n$-generic document:] 
Let $0 < \delta \leq 1$. We say that a document $d \in \calD$ is $\delta$-$n$-\emph{generic} if the following 
holds:
$$\forall i \geq \delta \sqrt{2 \ln n}:\qquad |\{t_j \in d \mid j \leq e^i\}| \geq (1-\delta) i.$$
\end{description}
}
\begin{definition}\label{def:reg}
Let $$\underbrace{t_1\;t_2}_{P_1}\quad\underbrace{t_3\;t_4\;t_5\;t_6\;t_7}_{P_2}\quad\ldots$$
be a partition of the set of terms. The group $P_i$ includes terms from $t_{\lceil e^{i-1} \rceil}$
to $t_{\lfloor e^i\rfloor}$. We say that a document $d\in \calD$ is \emph{regular} if it contains exactly $\ln m$ terms
$p_1\ldots p_{\ln m}$ such that $p_i \in P_i$.
\end{definition}
\begin{remark}
Note that the expected number of terms in each group $P_i$ is approximately one. If $\ln m$ is 
not an integer then the index of the last group is $\lceil \ln m \rceil$. 
In this case the expected number of terms in this group is smaller than one. To make the following proofs 
more legible we do not demand that the number of terms of a document or the matching size is an integer. 
In real settings these values have to be rounded appropriately. 
\end{remark}
\begin{definition}
Let $0 < \delta < 1$. We say that a document $d \in \calD$ is $\delta$-$n$-\emph{generic} if the following 
holds:
$$\forall i \geq \delta \sqrt{2 \ln n}:\qquad |\{t_j \in d \mid j \leq e^i\}| \geq (1-\delta) i.$$
\end{definition}

\begin{lemma}\label{lem1}
Let $0 < \delta < 1$ and $c=e^{-\delta^2/2}$. Let $d \in \calD$ be a random document following the Zipf model. 
It holds that for a sufficiently large $n$ the probability that $d$ is $\delta$-$n$-generic is 
greater than $1-c^{1.4\,\delta\sqrt{\ln n}}/(1-c)$.
\end{lemma}
\begin{proof}
Let $X$ be a random variable denoting the expected number of terms in $d$ up to the term $t_{e^i}$. 
For a fixed $i$ the Chernoff bound $P(X \leq (1-\delta)EX) \leq e^{-EX \delta^2/2}$ yield that the probability that 
$d$ contains less than $(1-\delta)i$ terms up to the term $t_{e^i}$ is smaller than $e^{-i \delta^2/2}$. This holds 
since $i < EX$ and therefore
$$P(X \leq (1-\delta)i) \leq P(X \leq (1-\delta)EX) \leq e^{-EX \delta^2/2} < e^{-i \delta^2/2}.$$
So the probability that $d$ is not $\delta$-$n$-generic is for large $n$ bounded by 
\begin{eqnarray*}
\sum_{i\geq0} c^i - \sum_{i=0}^{\lfloor\delta\sqrt{2 \ln n}\rfloor-1} c^i &=& \frac{1}{1-c} - 
\frac{1-c^{\lfloor\delta\sqrt{2\ln n}\rfloor}}{1-c}\\
&=& \frac{c^{\lfloor\delta\sqrt{2\ln n}\rfloor}}{1-c}\\
&\leq& \frac{c^{\sqrt{2}\,\delta\sqrt{\ln n}-1}}{1-c}\\
&<& \frac{c^{1.4\,\delta\sqrt{\ln n}}}{1-c}
\end{eqnarray*}
This holds because $1.4 < \sqrt{2}$ and $n$ is large. Overall, the probability that $d$ contains 
$(1-\delta)i$ or more terms is greater than $1-c^{1.4\,\delta\sqrt{\ln n}}/(1-c)$. 
\end{proof}

\begin{lemma}\label{lem2}
Let $d\in\calD$ be a random document following the Zipf model and let $0 < \delta < 1$ and $c = e^{-\delta^2/2}$. 
If we insert the first $\delta \sqrt{2 \ln n}$ missing terms to $d$ (assuming that there are missing terms), then for a 
sufficiently large $n$ the following holds:
\[P\left(\forall i \leq \sqrt{2 \ln n}:\> |\{t_j \in d \mid j \leq e^i\}| \geq i\right) > 1-c^{1.4\,\delta\sqrt{\ln n}}/(1-c).\]
\end{lemma}

\begin{proof}
Since we insert the first $\delta \sqrt{2 \ln n}$ missing terms to $d$, the  number of terms in $d$ is 
always at least $\delta \sqrt{2 \ln n}$. Thus, the statement holds trivially for $i < \delta \sqrt{2 \ln n}$. 
So let's consider the case $i \geq \delta \sqrt{2 \ln n}$. By Lemma \ref{lem1} we know that the probability 
that $d$ is $\delta$-$n$-generic is greater than $1-c^{1.4\,\delta\sqrt{\ln n}}/(1-c)$ for large $n$. 
Now, by inserting the first $\delta \sqrt{2\ln n}$ missing terms to $d$, we see that the probability 
that there are at least $i$ terms $t_j$ with $j \leq e^i$ is greater than 
$1-c^{1.4\,\delta\sqrt{\ln n}}/(1-c)$. 
\end{proof}

We now introduce a threshold, called \emph{\ml}, to give statements about the most
probable size of a maximal intersection:
$$q = q_n:= \sqrt{2 \ln n}.$$


\begin{theorem}[\ML\ for the Zipf
Model]\label{thm:zipf} Let $\calD_n = \{d_1,\ldots,d_n\}$ be a document collection following the
Zipf model. 
\begin{enumerate}
\item \emph{(Prefix match)}. Let $0 < \delta < 1$ be fixed. 
Let $\gamma = 2 + \delta\sqrt{2 \ln n}$ and $c=e^{-\delta^2/2}$. For sufficiently large $n,m$ the following holds: The probability 
that there exists a document in $\calD_n$ that contains the first $q-\gamma$ terms 
of a query document $d_\new \in \calD$ following the Zipf model 
is greater than $1 - c^{\delta\sqrt{\ln n}}/(1-c)$. Thus, the probability tends to one as $n \to \infty$.

\item \emph{(Any match)}. Let $\varepsilon>0$ be fixed.
The probability that there exists a document in $\calD_n$ that
contains more than $(1+\varepsilon)q$ \comment{arbitrary} terms of
a query document $d_\new \in \calD$ following the Zipf model tends to zero as $n\to\infty$.
\end{enumerate}
\end{theorem}

\medskip
\begin{proof}
\begin{enumerate}
\item 
Let $d_R$ be a fixed regular document (Definition \ref{def:reg}).
The probability that a document from $\calD_n$ contains the prefix of
length $q-2$ of $d_R$ is at least
\begin{eqnarray*}
\frac{1}{e}\cdot \ldots \cdot \frac{1}{e^{q-2}} \> >\>
\frac{1}{e^{(q -1)^2/2}} \> &=& \>
\frac{1}{e^{(q^2-2q+1)/2}}\\
\>&=&\> \frac{e^{q-1/2}}{n}. 
\end{eqnarray*}
Note that $e^{q-1/2} < n$ since $e^{(q -1)^2/2} > 1$.
This means that the probability that
there exists no document in $\calD_n$ that contains the $(q-2)$-prefix of $d_R$ is no more than
$$\left( 1 - \frac{e^{q-1/2}}{n} \right)^n
\> <\> e^{-e^{q-1/2}},$$ which follows from inequality
$(\ref{eq:1})$. 
So with probability greater than $1 - e^{-e^{q-1/2}}$
there exists a document in $\calD_n$ that has all terms from the
$(q - 2)$-prefix of $d_R$. Consider $d_\new$. If we insert the first $\delta \sqrt{2 \ln n}$ 
missing terms to $d_\new$, Lemma \ref{lem2} implies that for large $n$ the probability that $d_\new$ 
contains in every group $P_i$, $i \leq \sqrt{2\ln n}$, at least as many terms as $d_R$ is greater than $1-c^{1.4\,\delta\sqrt{\ln n}}/(1-c)$. 
Therefore, the probability that there exists a document in $\calD_n$ that matches the prefix of length $q-2$ of 
the extended query document $d_\new$ is greater than  
$$\left(1-c^{1.4\,\delta\sqrt{\ln n}}/(1-c)\right)\left(1-e^{-e^{q-1/2}}\right).$$
For large $n$, this product is at least $1-c^{\delta\sqrt{\ln n}}/(1-c)$.
It remains to notice that by removing the initially inserted $\delta \sqrt{2 \ln n}$ terms from 
$d_\new$ we still match $q - \gamma$ terms. This concludes the proof. 

\medskip

\item Let us fix a query document $d_\new$ for now.
Let $d$ be a random document following the Zipf model and let $\l > 0$. 
We can evaluate the Laplace transform of the intersection size as follows:
\begin{eqnarray*}
E \exp \left( \l |d_\new \cap d| \right) &=& 
\prod_{j:t_j\in d_\new} \left(1-\frac 1 j + \frac 1 j e^\l \right)
\le \prod_{j:t_j\in d_\new} \left(1 + \frac {e^\l} j \right)
\cr
&=& \prod_{{j:t_j\in d_\new \atop j>e^\l}} \left(1 + \frac {e^\l} j \right)
\cdot
\prod_{{j:t_j\in d_\new \atop j\le e^\l}} \frac {e^\l} j \left(1 + \frac j {e^\l}  \right)
.
\end{eqnarray*}
It follows that
\begin{eqnarray*}
\ln E \exp \left( \l |d_\new \cap d| \right) &\le&
\sum_{{j:t_j\in d_\new \atop j\ge e^\l}} \frac {e^\l} j
+
 \sum_{{j:t_j\in d_\new \atop j\le e^\l}}  (\l-\ln j)
 +
 \sum_{{j:t_j\in d_\new \atop j\le e^\l}} \frac j {e^\l}
\cr
&=& e^\l T_1 +\l T_2- T_3 + e^{-\l}T_4,
\end{eqnarray*}
where
\[
T_1 = \sum_{{j: t_j\in d_\new \atop j\ge e^\l}} \frac 1 j, \qquad T_2 = |d_\new \cap [1, e^\l]|, 
\]
\[
T_3 = \sum_{{j:t_j\in d_\new\atop j\le e^\l}} \ln j, \qquad T_4 = \sum_{{j:t_j\in d_\new\atop j\le e^\l}} j.
\]

We will assume that our fixed query $d_\new$ satisfies the following four {\it regularity conditions}.
\[
e^\l T_1\le \eps \l^2, \qquad  T_2\le (1+\eps) \l,
\]
\[
 T_3\ge (1-\eps)\frac{ \l^2} 2, \qquad  e^{-\l} T_4\le \eps \l^2.
\]

Under these regularity conditions we obtain
\[
   \ln E \exp \left( \l |d_\new \cap d| \right) \le (1+7\eps) \frac{ \l^2} 2\ .
\]

Assuming $(d_j)$ to be a sample of $n$ independent documents distributed according to the Zipf model,  
by the Chebyshev exponential inequality, for any $r>0$ we have
\begin{align*}
P\left(\max_{j\le n} |d_\new \cap d_j| \ge r\right) &\le
  n\  P\left( |d_\new \cap d| \ge r\right) \\
&\le n\ \frac {E \exp \left( \l |d_\new \cap d| \right)}{e^{\l r}} \\
&\le n \exp \left( (1+7\eps) \frac{ \l^2} 2 -\l r \right).
\end{align*}
Take any $\gamma>0$. By choosing now
\[
r=(\sqrt{2\ln n}+\gamma) (1+7\eps)
\]
and
\[
\l=\sqrt{2\ln n}+\gamma
\]
we obtain $(1+7\eps)\l^2=\l r$, hence
\begin{align*}
P\left( \max_{j\le n} |d_\new \cap d_j| \ge r\right) &\le
  n\ \exp \left( -(1+7\eps) \frac{ \l^2} 2  \right) \\
&\le n\ \exp \left( - \frac{ (\sqrt{2\ln n}+\gamma )^2} 2  \right)
\to 0
\end{align*}
for $n\to \infty$, as required. 
Let now  the query $d_\new$ be randomly chosen according to the Zipf model. 
For a sufficiently large $\l$ it is true that
\begin{eqnarray*}
ET_1 &=& \sum_{j\ge e^\l} \frac 1 {j^2} \le e^{-\l};
\cr
ET_2 &=& \sum_{j\le e^\l} \frac 1 {j} \le \l;\qquad
\Var T_2 \le \sum_{j\le e^\l} \frac 1 {j} \le \l;
\cr
ET_3 &=& \sum_{j\le e^\l} \frac {\ln j} {j} \le \frac{\l^2} 2;\qquad
\Var T_3 \le \sum_{j\le e^\l} \frac {\ln^2 j} {j} \le \frac{\l^3} 3;
\cr
ET_4 &=& \sum_{j\le e^\l} 1 \le e^{\l}.
\end{eqnarray*}
Let us explain the bound for $ET_3$. Let $2\leq b \leq S$ be integers.  
$$\sum_{j \le S} \frac {\ln j} {j} \leq \int_{b}^{S} \frac{\ln x}{x} dx + \sum_{j=2}^{b} \frac{\ln j}{j} = \frac{\ln^2 S}{2} - \frac{\ln^2 b}{2} + \sum_{j=2}^{b} \frac{\ln j}{j}.$$
For $b \geq 21$ it holds that $$\frac{\ln^2 b}{2} \geq \sum_{j=2}^{b} \frac{\ln j}{j}$$ and therefore $$\sum_{j \le S} \frac{\ln j}{j} \leq \frac{\ln^2 S}{2}.$$ 
Hence, $\l \geq 4$ yields the desired bound on $ET_3$. The bound on $\Var T_3$ is shown by similar techniques.

It remains to notice that the probability of each of the four regularity conditions to be true
for $d_\new$ tends to one as $\lambda\to\infty$. Indeed, the expectations and variances calculated above easily show
this fact.
\end{enumerate}
\end{proof}

By Theorem \ref{thm:zipf} we can conclude that with high
probability there exists a~document in $\calD$ that matches the prefix of length $q - \gamma$ of $d_\new$, 
whereas the probability to find a document that has more than $(1+\varepsilon)q$
common terms with $d_\new$ (at arbitrary positions) is quite small. 
Therefore, it suffices to determine a document that has a maximal common prefix 
with the query document. This fact, however, allows to sort the documents 
according to their sorted term lists\footnote{In a sorted term list 
the terms of a document are ordered according to the position of the terms in the frequency 
table.} and then perform a binary search based on the sorted 
term list of the query document (Figure \ref{fig:zipfalgo}).
\begin{figure}[h]\label{fig:zipfalgo}
\begin{BoxItAlgo}
  \item[{\bf Preprocessing}] \rule{0pt}{0pt}
    \begin{enumerate}
      \item For every document: Sort the term list according to
        the position of the terms in the frequency table.
      \item Sort the documents according to their sorted term lists. 
    \end{enumerate}
  \item[{\bf Query}]
  Find a document having the maximal common
  prefix with the query document by binary search.
\end{BoxItAlgo}
\caption{{\bf \maxint\ algorithm in the Zipf model}}
\end{figure}
The running time is as follows (for the \emph{average}
case\footnote{Only for the average case our constraints from
Section \ref{sec:intro} are preserved.} analysis we assume that
the length of term lists is $\log m \in \calO(\log n)$,
for the \emph{worst} case analysis the length is $m$):

\bigskip
\comment{
\renewcommand{\arraystretch}{1.3}
\begin{tabular}[ht]{|p{1.5cm}|p{3.5cm}|p{2.5cm}|}
  \hline
   & average & worst\\
  \hline
  Step 1 & $\calO(n \cdot \log m \cdot \log \log m)$ & $\calO(n \cdot m \cdot \log m)$\\
  Step 2 & $n^{1+o(1)}$ & $\calO(n^2 \cdot \log n)$\\
  Step 3 & $\calO(\log m \cdot n \cdot \log n)$ & $\calO(m \cdot n \cdot \log n)$\\
  Query & $\calO(\log^2 n)$ & $\calO(\log^2 n)$\\
  \hline
\end{tabular}

\bigskip
Let us explain the estimations from the second line. The number of all
possible regular $(q-\gamma)$-lists is equal to
$$|P_1|\cdot\dots\cdot|P_{q-\gamma}|\leq \prod_{k=1}^{q-\gamma} e^k < e^{q^2/2}=n$$
Therefore, a single document can generate at most $n$ different
regular $(q-\gamma)$-lists, the $\log n$ factor arises from the size of a
single list. Let us
prove the bound for the average case. The probability of containing
some fixed regular $(q-\gamma)$-list is $n^{-1+o(1)}$. Summing over all
possible lists we see that the expected number of generated regular
lists per document is at most $n^{-1+o(1)}\cdot n = n^{o(1)}$.
Therefore, the expected time for the indexing stage is $n^{1+o(1)}$.
}
\renewcommand{\arraystretch}{1.3}
\begin{tabular}[ht]{|p{1.5cm}|p{3.5cm}|p{2.5cm}|}
  \hline
   & average case & worst case\\
  \hline
  Step 1 & $\calO(n \cdot \log n \cdot \log \log n)$ & $\calO(n \cdot m \cdot \log n)$\\
  Step 2 & $\calO(n\cdot \log^2 n)$ & $\calO(n \cdot \log n \cdot m)$\\
  Query & $\calO(\log^2 n)$ & $\calO(\log n \cdot m)$\\
  \hline
\end{tabular}

\bigskip

The $\log$ factor in the query step is due to the fact that the algorithm 
performs a binary search on the document collection. Since in the average case 
the length of a term list is $\calO(\log n)$, we get another $\log$ factor resulting in 
a query time of $\calO(\log^2 n)$. One can try to improve the accuracy of our algorithm by finding a
``maximal prefix with at most one difference to the query
document''. A recent technique called ``indexing with errors"
\cite{CGL04,MN05} might be useful for such an extension.

\section{\maxint\ in the Hierarchical Scheme}\label{sec:HierSch}
\begin{figure}
\centering{
\begin{tikzpicture}[scale=.06,line width=1pt]

 \draw (0,0) rectangle (80,80);
 \draw (0,20) -- (80,20);
 \draw (0,40) -- (80,40);
 \draw (0,60) -- (80,60);

 \draw (40,60) -- (40,0);
 \draw (20,40) -- (20,0);
 \draw (60,40) -- (60,0);

 \draw (10,20) -- (10,0);
 \draw (30,20) -- (30,0);
 \draw (50,20) -- (50,0);
 \draw (70,20) -- (70,0);

 \filldraw (38,67) circle (10pt);
 \filldraw (40,75) circle (10pt);
 \filldraw (42,70) circle (10pt);

 \filldraw (18,47) circle (10pt);
 \filldraw (20,55) circle (10pt);
 \filldraw (22,50) circle (10pt);
 \filldraw (58,47) circle (10pt);
 \filldraw (60,55) circle (10pt);
 \filldraw (62,50) circle (10pt);

 \filldraw (8,27) circle (10pt);
 \filldraw (10,35) circle (10pt);
 \filldraw (12,30) circle (10pt);
 \filldraw (28,27) circle (10pt);
 \filldraw (30,35) circle (10pt);
 \filldraw (32,30) circle (10pt);
 \filldraw (48,27) circle (10pt);
 \filldraw (50,35) circle (10pt);
 \filldraw (52,30) circle (10pt);
 \filldraw (68,27) circle (10pt);
 \filldraw (70,35) circle (10pt);
 \filldraw (72,30) circle (10pt);

 \filldraw (2,7) circle (10pt);
 \filldraw (5,15) circle (10pt);
 \filldraw (7,10) circle (10pt);
 \filldraw (12,7) circle (10pt);
 \filldraw (15,15) circle (10pt);
 \filldraw (17,10) circle (10pt);
 \filldraw (22,7) circle (10pt);
 \filldraw (25,15) circle (10pt);
 \filldraw (27,10) circle (10pt);
 \filldraw (32,7) circle (10pt);
 \filldraw (35,15) circle (10pt);
 \filldraw (37,10) circle (10pt);
 \filldraw (42,7) circle (10pt);
 \filldraw (45,15) circle (10pt);
 \filldraw (47,10) circle (10pt);
 \filldraw (52,7) circle (10pt);
 \filldraw (55,15) circle (10pt);
 \filldraw (57,10) circle (10pt);
 \filldraw (62,7) circle (10pt);
 \filldraw (65,15) circle (10pt);
 \filldraw (67,10) circle (10pt);
 \filldraw (72,7) circle (10pt);
 \filldraw (75,15) circle (10pt);
 \filldraw (77,10) circle (10pt);

 \draw[thick, line width=2pt] (0,60) rectangle (80,80);
 \draw[thick, line width=2pt] (40,40) rectangle (80,60);
 \draw[thick, line width=2pt] (40,20) rectangle (60,40);
 \draw[thick, line width=2pt] (50,0) rectangle (60,20);

 \draw[thick, line width=2pt] (40,75) circle (60pt);
 \draw[thick, line width=2pt] (48,27) circle (60pt);
 \draw[thick, line width=2pt] (62,50) circle (60pt);
 \draw[thick, line width=2pt] (57,10) circle (60pt);

\end{tikzpicture}
}
\caption{Hierarchical scheme}\label{fig:hs}
\end{figure}
Our second model is motivated by the observation that in many existing applications terms
can be ordered hierarchically. Let $k \geq 8$ be an integer and let $\calT$ be a set of $(2^k-1) \cdot k$
different terms. A document collection $\calD_{k}$ consists of $2^k$ sets 
where every set $d \in \calD_{k}$ is a subset of $\calT$ with $|d| = k$. 
A \emph{hierarchical scheme} is a table with $k$ levels,
level $1$ to level $k$. Level $i$, $1 \leq i \leq k$, is divided into $2^{i-1}$ cells, cell $C_{i,1}$ to 
cell $C_{i,2^{i-1}}$. For $2 \leq l \leq k$ we say that cell $C_{l-1,j}$, $1\leq j \leq 2^{l-2}$, is \emph{above} 
cell $C_{l,j'}$, $1\leq j' \leq 2^{l-1}$, if $\lceil j'/2 \rceil = j$. Every cell 
contains $k$ terms. A document collection based on this scheme can be generated as follows: 
Every document is generated independently. Choose a random cell on level $k$ and mark it. Then, for $l = \{k,\ldots,2\}$, 
mark on level $l-1$ the cell that is above the already marked cell on level $l$. 
Now choose one random term in every marked cell. 
The so defined set of terms form a document of our collection. Note
that every document generated by this process corresponds to a unique sequence of 
cells. We'll call such a sequence a \emph{cell path}. There exists a natural ordering 
on these cell paths where the cells $C_{1,1},C_{2,1},\ldots,C_{k,1}$ describe the leftmost 
path, the cells $C_{1,1},C_{2,2},\ldots,C_{k,2^{k-1}}$ accordingly the rightmost one (see Figure \ref{fig:hs}).

\begin{remark}
We claim that the hierarchical scheme follows \emph{Zipf's law}.
To be more precise, the following holds: For every level the
product of expected frequency and expected frequency rank of a term is the
same. Indeed, the expected frequency of a term on level $i$ is
given by the formula $2^k/(2^{i-1} \cdot k)$. The
expected rank of a~term is given by the formula $(2^{i-1}
-1) \cdot k + 2^{i-2} \cdot k$. Hence, the product between frequency
and frequency rank (divided by $2^k$) is equal to
$$\frac{2^k}{2^{i-1} \cdot k} \cdot \left(\frac{3}{2} \cdot 2^{i-1} -1\right)\cdot \frac{k}{2^k} \in \left[0.5, 1.5\right),$$
which means it lies in a fixed interval and therefore follows the idea of Zipf's law.

\end{remark}

This time we introduce two \emph{\mls} to give statements about
the most probable size of a maximal intersection. The \mls\
are
$$q \>=\> \frac{k}{1+\log k}\qquad\mbox{and}\qquad q'\>=\> \frac{k}{\log k}.$$
\begin{remark}
Again, to keep proofs more legible, we do not demand that the length of 
a prefix or the matching size is an integer (except part one of the following theorem). But clearly, in real 
settings these values have to be rounded appropriately. 
\end{remark}
\begin{theorem}[\MLs\ for Hierarchical Scheme]\label{thm:hier}
Let $k\geq 8$ be an integer and $2 \leq \gamma<q$. Let $\calD_{k}$ be a document 
collection following the hierarchical scheme.
\begin{enumerate}
\item \emph{(Prefix match)}. The probability that there exists a document $d \in \calD_{k}$
that matches the first $\lfloor q-\gamma\rfloor$ terms (i.e. the terms from level $1$ to level $\lfloor q-\gamma\rfloor$) of a 
query document $d_\new$ following the hierarchical scheme is greater than $1-2^{-(2k)^{\gamma}}$.
\item
\emph{(Any match)}. The probability that there exists a document $d \in \calD_{k}$ that
matches at least $q'+\gamma$ \comment{arbitrary} terms of 
a query document $d_\new$ following the hierarchical scheme is smaller than $2/k^{\gamma-1}$.
\end{enumerate}
\end{theorem}
\medskip
\begin{proof}
\begin{enumerate}
\item The number of different prefixes of length $\lfloor q-\gamma \rfloor$ is at most
$$k(2k)^{q-\gamma-1} \><\> 2^{(1 + \log k)(q - \gamma)} \>=\> 2^{(1 + \log k)(k/(1 + \log k) - \gamma)}
\>=\> 2^k\cdot(2k)^{-\gamma}.$$ So the probability that a new
document does not match a prefix of length $\lfloor q - \gamma \rfloor$ with any
document from $\calD_{k}$ is smaller than
$$\left(1 - \frac{(2k)^{\gamma}}{2^k}\right)^{2^k} \><\> e^{-(2k)^{\gamma}}\><\> 2^{-(2k)^{\gamma}}.$$
For $k\geq 8$ and $\gamma < q$ it holds that $(2k)^{\gamma} < 2^k$ and therefore the above inequality follows 
from inequality $(\ref{eq:1})$. We get that the
probability that there exists a document in $\calD_{k}$ with the same prefix as
$d_\new$ of length $\lfloor q - \gamma \rfloor $ is greater than $1 -
2^{-(2k)^{\gamma}}$.
\medskip
\item Let $t \geq q'+\gamma$ be the last position where the terms
of $d$ and $d_\new$ match. We want to estimate the probability
that $d_\new$ matches at least $q'+\gamma$ terms at arbitrary
positions with $d$. The probability that the first $t$ terms
(beginning at level $1$) of $d$ and $d_\new$ are all in the same
cells is $2^{1-t}$. The probability that at least $q'+\gamma$
terms are matched on some fixed positions is no more than
$\left(1/k\right)^{q'+\gamma}\cdot
\left((k-1)/k\right)^{t - q' -\gamma}$. An
upper bound for the number of different possibilities of matching
at least $q'+\gamma$ out of $t$ terms is $2^t$. Since the factor
$\left((k-1)/k\right)^{t - q' - \gamma}$ is smaller than 1,
overall we get that the probability that $d_\new$ matches at least
$q'+\gamma$ terms at arbitrary positions with $d$ is smaller than
$$k \cdot 2^t \cdot \left(\frac{1}{k}\right)^{q'+\gamma}\cdot 2^{1-t} \>=\> 2 \cdot k \cdot \left(\frac{1}{k}\right)^{q'+\gamma} \>=\> 2 \cdot \left(\frac{1}{k}\right)^{q'+\gamma-1}.$$
The factor $k$ in the above equation is due to the fact that we need to consider all possible levels for the last matched position $t$.
Now the probability that no document matches at $q'+\gamma$ arbitrary positions with $d_\new$
is at least
\begin{multline*}
 \left(1 - 2 \cdot \left(\frac{1}{k}\right)^{q'+\gamma-1}\right)^{2^k} \>=\>
 \left(1 - \frac{1}{2^k\cdot \frac{k^{\gamma-1}}{2}}\right)^{2^k} \>
 \geq\> 1 - \frac{2}{k^{\gamma-1}},
\end{multline*}
which follows from inequality $(\ref{eq:2})$, since $\gamma \geq 2$ and $k\geq
8$. So the probability that there exists a document in $\calD_{k}$ that matches 
at least $q'+\gamma$ terms of $d_\new$ is smaller
than $2/k^{\gamma-1}$.
\end{enumerate}
\end{proof}

Theorem \ref{thm:hier} yields that also for the hierarchical scheme it suffices to search 
a document that has a maximal common prefix with the query document. The resulting algorithm is analogue to the one
for the Zipf model and summarized on Figure \ref{fig:hieralgo}.

\begin{figure}[h]
\begin{BoxItAlgo}
  \item[{\bf Preprocessing}] \rule{0pt}{0pt}
    \begin{enumerate}
      \item For every document: Sort the term list according to
        the hierarchical scheme, i.e. according to the levels in which
        the terms appear.
      \item Sort the documents according to their corresponding cell paths, i.e.
        documents that correspond to the leftmost path in the scheme are at the beginning
    of the sorted list.
        Documents that correspond to the same cell path are sorted
        lexicographically.
    \end{enumerate}
  \item[{\bf Query}]
  Find a document having the maximal common
  prefix with the query document by binary search.
\end{BoxItAlgo}
\caption{{\bf \maxint\ algorithm in the hierarchical scheme}}\label{fig:hieralgo}
\end{figure}

The running time is shown in the table below. Note that for the hierarchical scheme we 
only perform a worst case analysis since every document has equal length. 

\bigskip
\renewcommand{\arraystretch}{1.3}
\begin{tabular}[ht]{|p{1.5cm}|c|p{2.5cm}|p{2.5cm}|}
  \hline
   & average case & worst case\\
  \hline
  Step 1 & -- & $\calO(2^k \cdot k \cdot \log k)$\\
  Step 2 & -- & $\calO(2^k \cdot k^2)$\\
   Query & -- & $\calO(k^2)$\\
  \hline
\end{tabular}

\section{Further Work}\label{sec:open}

In this paper we have shown that assumptions on the random nature
of the input can lead to \emph{provable} time and accuracy bounds
for \maxint. Also, we have discovered a \maxint\ threshold
phenomenon in two randomized models.

The next step is to understand it better. Does it hold for other
randomized models from \cite{N03}, especially for generalized
random graphs with a power-law degree distribution? Does it hold
in the real life networks? Can we introduce randomized models for
sparse vector collections and find a similar effect there? It was
observed that tractable instances of nearest neighbors have small
intrinsic dimension \cite{C06,GLS08,KL04,KR02}. Does the same
effect hold for the Zipf model and the hierarchical scheme? Of course,
the most challenging problem is to find an exact algorithm for
\maxint\ (preserving our time constraints) or to prove its
hardness. What are other particular cases or assumptions that have
efficient \maxint\ solutions? On the other hand, we have a very
particular subcase for which we still do not believe in a positive
solution. Hence, we ask for a hardness proof for the following
\emph{on-line inclusion problem}.
\begin{BoxIt}[On-line Inclusion Problem]
  \DB A family $\calF$ of $2^k$ subsets of $[1\ldots k^2]$.
  \QY Given a set $f_\new\subseteq [1\ldots k^2]$, decide whether
    there exists an $f\in\calF$ such that $f_\new\subseteq f$.
  \CONSTR Space for preprocessed data $2^k\cdot \poly(k)$.\\
    Query time $\poly(k)$.
\end{BoxIt}
Note that we have a constraint on space for preprocessing, not
time. A related problem but with a much stronger restriction on
preprocessing space was proven to be hard by Bruck and Naor
\cite{BN90}.

Our algorithm in Section \ref{sec:HierSch} uses  polylogarithmic
time (in the number of documents) but it returns only an
approximate solution with high probability (not every time). Can
we get an optimal solution or at least a \emph{guaranteed}
approximation by relaxing the time constraint to \emph{expected}
polylogarithmic time?

The maximal intersection problem is a special case of a whole
family of problems called \emph{Strongest Connection Problems}
(SCP) which covers all problems fitting the following framework.
Consider some class of graphs $\calG$ and some class of paths
$\calP$.
\begin{BoxIt}[Strongest Connection Problems]
  \DB A~graph $G\in\calG$.
  \QY Given a (new) vertex $v$ (together
    with edges connecting $v$ with $G$),
    return a vertex $u\in G$
    with a maximal number of $\calP$-paths from $v$ to $u$.
  \CONSTR Preprocessing time $o(|G|^2)$.\\
    Query time $o(|G|)$.
\end{BoxIt}

A number of well-motivated problems fall into the family of SCP.
Some of them are listed below but the application of the SCP
framework is not limited to these instances.
\begin{description}
  \item[Recommendations.]
    $\calG$: bipartite graphs;
    $\calP$: paths of length three. \\
    \emph{Example:} A graph $G$ is partitioned into vertices
    representing people and books, and the edge relation
    describes who has bought which books. Given a person $v$. Which is
    a book (not purchased by $v$) that is most often bought by people
    that had been interested in books of $v$?
  \item[Similarity in folksonomies.]
    $\calG$: tripartite 3-graphs where the set of vertices is partitioned
      into $(V_0, V_1, V_2)$ and $E$ is an edge relation
      in $V_0\times V_1\times V_2$;
    $\calP$: paths consisting of two edges that overlap either in $V_0$ or in $V_1$ or both
      in $V_0$ and $V_1$. \\
    \emph{Example:} A graph representing all events of kind ``a user $U$ used
    a tag $T$ to label a website $W$''. Then, the strongest connection
    query for the tag $T_{new}$ returns another tag that was
    most often co-used by the same users and/or applied to the same
    websites. Such tripartite 3-graph is accumulated in
    \mbox{\url{del.icio.us}}.
  \comment{
  \item[Semantic search.]
    $\calG$: graphs $(V, E)$ whose edge relations $E$ are $k$-partitioned
      into $(E_1, \ldots, E_k)$ where $E_i\subseteq V_{i-1}\times V_i$
      with $1\leq i\leq k$ and $V=\bigcup_{0\leq j\leq k}V_j$. \\
    \emph{Example:} A graph relating friends ($V_0=V_1$), bars ($V_2$),
    and drinks ($V_3$).
    Given a person $v\in V_0$. Which is a most popular drink that is
    available in bars visited by the friends of $v$?
  }
\end{description}

\paragraph{Acknowledgments.} 
We thank the anonymous referees for their very detailed reports. We also thank Volker Diekert, J\"urn Laun, and 
Ulrich Hertrampf for many helpful comments.

\bibliographystyle{abbrv}
\bibliography{References}

\end{document}